\documentclass[aps,prl,twocolumn,superscriptaddress,floatfix,nofootinbib,  raggedbottom]{revtex4-2}
\usepackage{amsmath,amssymb,bm,braket,graphicx,hyperref,xcolor,amsthm}

\newtheorem{theorem}{Theorem}
\newtheorem{definition}{Definition}

\begin{document}

\title{Reshaping quantum annealing landscapes with diagonal catalysts}

\author{Andrés N. Cáliz}
\affiliation{Qilimanjaro Quantum Tech, 08019 Barcelona, Spain}
\affiliation{Departament de Física, Universitat de Barcelona, 08007 Barcelona, Spain}

\author{Carlos Ramon-Escandell}
\affiliation{Qilimanjaro Quantum Tech, 08019 Barcelona, Spain}

\author{Finnley Paolella}
\affiliation{Qilimanjaro Quantum Tech, 08019 Barcelona, Spain}

\author{Josep Bosch}
\affiliation{Qilimanjaro Quantum Tech, 08019 Barcelona, Spain}
\affiliation{Universitat Politècnica de Catalunya, Carrer de Jordi Girona 1-3, 08034 Barcelona, Spain}

\author{Jan Nogué}
\affiliation{Qilimanjaro Quantum Tech, 08019 Barcelona, Spain}
\affiliation{Universitat Politècnica de Catalunya, Carrer de Jordi Girona 1-3, 08034 Barcelona, Spain}

\author{Arnau Riera}
\affiliation{Qilimanjaro Quantum Tech, 08019 Barcelona, Spain}

\author{Jordi Riu}
\email{jordi.riu@qilimanjaro.tech}
\affiliation{Qilimanjaro Quantum Tech, 08019 Barcelona, Spain}

\date{\today}

\begin{abstract}
Quantum annealing is often limited by population trapped in local minima many spin flips from the solution. We introduce a mathematical framework to understand the connection between energy and Hamming distance in optimization problems. Using this, we build ZZ-catalysts from ground-state patterns of small frustration-free subproblems that make configurations far from the solution less energetically competitive. On sparse problems they multiply the near-solution probability at short sweeps, with gains persisting on fully-connected models and tunable via subproblem choice.
\end{abstract}

\maketitle

\emph{Introduction}---Quantum annealing (QA) is an optimization algorithm that encodes the solution of a classical optimization problem in the ground state of a cost Hamiltonian and attempts to reach it by adiabatic evolution from a simple quantum driver Hamiltonian~\cite{Kadowaki1998,Farhi2000,Santoro2002,Albash2018,Hauke2020}. The target is a Quadratic Unconstrained Binary Optimization (QUBO) cost, which for \(n\) spin variables \(z_i=\pm1\) takes the form
\begin{equation}
  H_P=\sum_i h_i Z_i+\sum_{i<j}J_{ij}Z_iZ_j,
  \label{eq:qubo}
\end{equation}
with \(Z_i\) diagonal Pauli operators. A vast range of NP-hard problems map directly onto this class~\cite{Lucas2014}. In this work, we restrict ourselves to the field-free pairwise case $H_P = \sum_{i<j}J_{ij}Z_iZ_j$, with $h_i=0$, which already contains the essential difficulty. The annealer interpolates linearly between the transverse-field driver and the problem,
\begin{equation}
  H(s)=-(1-s)\sum_i X_i+s\,H_P,\qquad s:0\to1,
  \label{eq:schedule}
\end{equation}
starting from the ground state of \(-\sum_iX_i\), namely the uniform superposition in the computational basis, and ending at the ground state of \(H_P\). By the adiabatic theorem, remaining in the instantaneous ground state requires an evolution time $T$ that scales inversely with the square of the minimum gap to a dynamically accessible excited state, $T\gtrsim\Delta_{\min}^{-2}$, up to prefactors and matrix-element dependence \cite{Jansen2007}. On hard instances, this gap can become exponentially small, making fully adiabatic evolution impractical ~\cite{Jorg2008,Young2010,Altshuler2010}. Finite-time anneals therefore often operate diabatically, where population can be redistributed at avoided crossings, i.e., points where eigenstate branches approach and are separated only by a small gap. Such crossings can occur across a broad energy range, rather than only between the ground and first excited states.

Avoided crossings arise between states localized in distinct basins of the cost landscape. When two such states are separated by a large Hamming distance $d_H=D$, their coupling emerges only at $D$th order in the transverse field, producing a perturbatively small gap,
\begin{equation}
  \Delta_{\mathrm{ac}}
  \sim
  \Gamma
  \left(\frac{\Gamma}{\Delta E}\right)^{D-1}
  =
  \Gamma e^{-(D-1)\ln(\Delta E/\Gamma)},
  \label{eq:gapD}
\end{equation}
with $\Gamma$ the transverse-field scale and $\Delta E$ a typical intermediate energy \cite{Amin_2008}. A sequence of such crossings can redistribute population throughout the energy spectrum.

Two broad strategies mitigate the resulting diabatic losses. Counterdiabatic (CD) driving~\cite{Berry2009,GueryOdelin2019} augments the driver with the adiabatic gauge potential, cancelling diabatic transitions at any sweep rate. In practice the gauge potential is known only as a series of nested commutators~\cite{SelsPolkovnikov2017}, analytically intractable for generic instances, and its leading operators for an Ising annealer ($Y_iZ_j$ and higher) are off-diagonal, many-body, and not implementable on near-term hardware. The second strategy adds a catalyst Hamiltonian, active only during the sweep~\cite{Farhi2002,Crosson2020}. Off-diagonal, non-stoquastic \(XX\)-type catalysts can open gaps~\cite{Seki2012,Hormozi2017,Albash2019,Takada2021,Feinstein2024,Ghosh2024,Nutricati2024}, but emulating them requires dedicated gadgets~\cite{Banks2025}. Purely diagonal catalysts instead reshape the classical cost using native longitudinal controls and have attracted recent interest~\cite{AlbashKowalsky2021,Hattori2026}. Their reach, however, has seemed fundamentally limited. Local-field catalysts enlarge the minimum gap exponentially only when the added field points toward a configuration a few spin flips away from the true ground state~\cite{AlbashKowalsky2021}.

In this Letter, we construct two-body diagonal catalysts from the coupling structure of the problem alone, with no knowledge of the solution, that concentrate the final population not only at low energy but also close to the solution in Hamming distance.\newline

\emph{Theoretical framework and catalyst design}---
We begin by organizing the \(2^n\) configurations according to their distance from a reference configuration.

\begin{definition}[Hamming shell]
For a reference configuration \(z^\star\in\{\pm1\}^n\), the Hamming shell of radius \(d\) is
\begin{equation}
    S_{d,z^\star}
    =
    \left\{
    z\in\{\pm1\}^n:
    d_H(z,z^\star)=d
    \right\}.
    \label{def:hamming_shell}
\end{equation}
\end{definition}
The notation \(\braket{\cdot}_{d,z^\star}\) and \(\operatorname{Var}_{d,z^\star}(\cdot)\) denote the uniform mean and variance over this shell. Let \(\mathcal G\) denote the set of computational-basis ground states. The distance of a configuration \(z\) to its closest state in $\mathcal{G}$ is then
\begin{equation}
\delta(z)=\min_{g\in\mathcal G} d_H(z,g).
\end{equation}
For the field-free problem considered here, global-spin-flip symmetry implies that \(-z^\star\) is also a ground state whenever \(z^\star\) is. In the absence of additional degeneracies,
\(\mathcal G=\{\pm z^\star\}\) and
\(
\delta(z)
=\min\{d,n-d\} \leq \frac{n}{2}.
\)

Under \(H_P\), the energy bands associated with different Hamming shells overlap Fig.~\ref{fig:catalyst_ordering}(a). Distant configurations can consequently undercut nearby ones, so low energy need not imply proximity to the solution. Such configurations can belong to competing basins associated with perturbative crossings; shell overlap provides a coarse-grained measure of this competition without requiring individual minima to be identified.

In the ideal distance-only limit, shells are perfectly ordered and no distant configuration undercuts a nearer one. Our goal is therefore to construct a diagonal catalyst \(C\) that brings \(H_P+C\) closer to this limit Fig.~\ref{fig:catalyst_ordering}(b), while ideally preserving the ordering induced by \(H_P\) among configurations within the same shell.
\begin{figure}[t]
    \centering
    \includegraphics[width=0.9\linewidth]{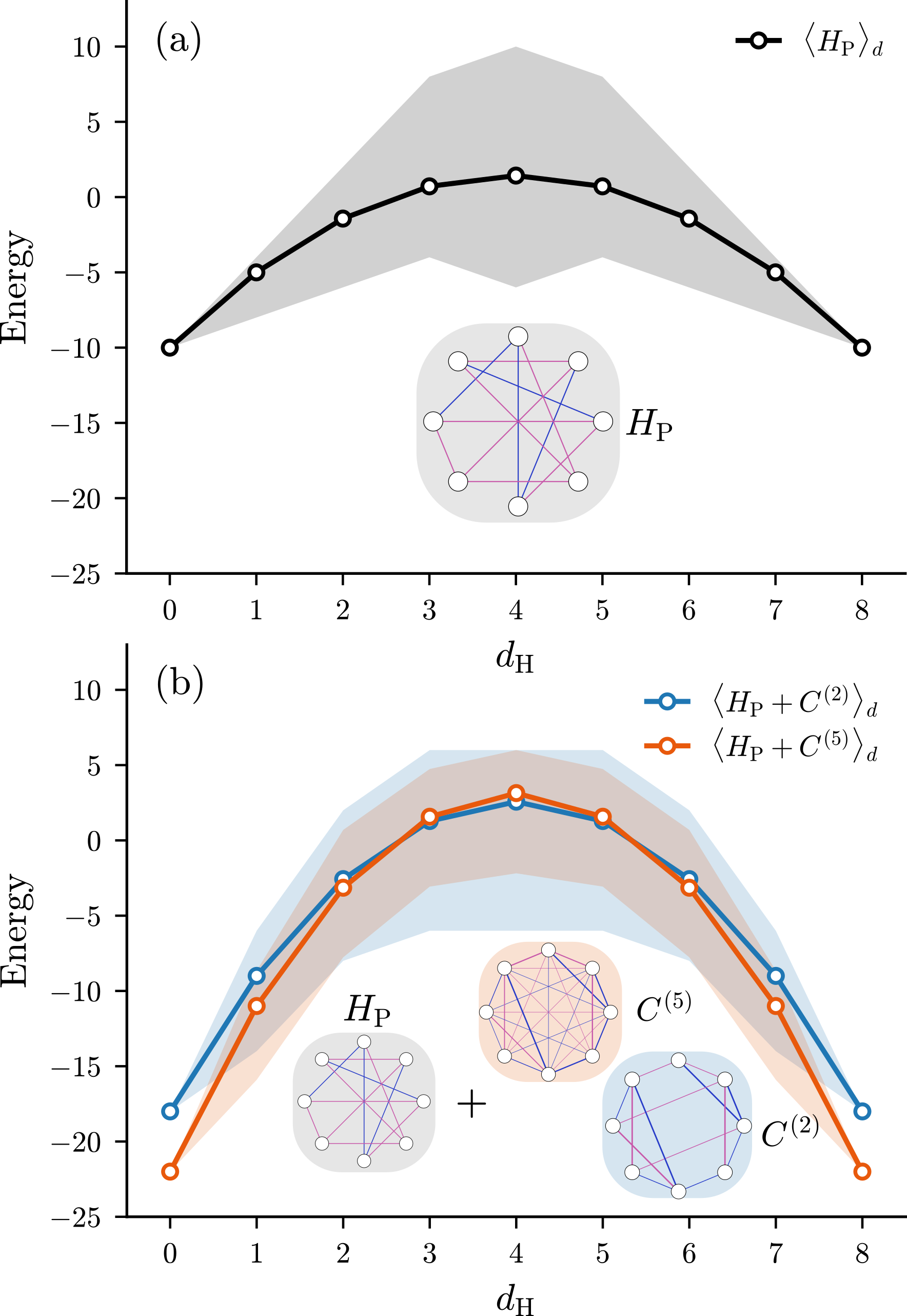}
    \caption{\textbf{Effect of the diagonal catalyst.} Shell mean energies
(open markers) and full within-shell energy range (shaded bands) vs.\
Hamming distance to a ground state \(z^\star\), for a signed 3-regular instance (\(n=8\)). (a)~Bare problem: the means trace the parabola of Eq.~\eqref{eq:parabola}, but the bands are broad and overlap, making configurations many flips from \(z^\star\) undercut close ones (see \(d_H=4\)) and providing low-energy candidates for competing basins. (b)~Open-path catalysts of order 2 (\(C^{(2)}\), blue) and order 5 (\(C^{(5)}\), orange) rescale the mean curve as dictated by Eq.~\eqref{eq:combinedparabola} and shrink the bands, improving the energy--distance correspondence; both effects grow with order. Insets: coupling graphs of \(H_P\), \(C^{(2)}\), and \(C^{(5)}\); blue edges denote \(J_{ij}>0\), red edges \(J_{ij}<0\), and thickness is proportional to coupling strength.}
    \label{fig:catalyst_ordering}
\end{figure}

\begin{theorem}[Hamming-shell moments]
Let
\begin{equation}
    H=\sum_{i<j}W_{ij}Z_iZ_j
\end{equation}
be any field-free two-body Ising operator on \(n\geq4\) spins, with \(W\) the corresponding symmetric zero-diagonal coupling matrix. For any reference configuration \(z^\star\) and any \(d\in\{0,\ldots,n\}\), its exact mean and variance over \(S_{d,z^\star}\) are
\begin{align}
    \braket{H}_{d,z^\star}
    &=
    \mu_2(d)H(z^\star),\qquad \mu_2(d) = \frac{(n-2d)^2-n}{n(n-1)},
    \label{eq:parabola}
    \\
    \operatorname{Var}_{d,z^\star}(H)
    &=
    S_2\!\left[1-\mu_4(d)\right] + U\!\left[\mu_2(d)-\mu_4(d)\right]
    \notag\\
    &\quad+ \left[\mu_4(d)-\mu_2(d)^2\right]H(z^\star)^2,
    \label{eq:variance}
\end{align}
where
\begin{equation}
S_2=\sum_{i<j}W_{ij}^2, \qquad U=\left\|Wz^\star\right\|^2-2S_2,
\end{equation}
and
\begin{equation}
    \mu_4(d) = \frac{(n-2d)^4-(6n-8)(n-2d)^2+3n^2-6n}{n(n-1)(n-2)(n-3)}.
    \label{eq:fourmoment}
\end{equation}
\label{th:1}
\end{theorem}
A proof via the hypergeometric moments of a uniformly sampled shell is given in Appendix~\ref{app:shell}.

The coefficient \(\mu_2(d)\) is purely combinatorial and contains no instance-specific information. When \(z^\star\) is a ground state, \(H(z^\star)<0\), and Eq.~\eqref{eq:parabola} traces a concave parabola with minima at \(d=0\) and \(d=n\), as seen in Fig.~\ref{fig:catalyst_ordering}(a). The shell means are therefore correctly ordered for every instance, and any overlap between shells arises from their within-shell spread.

Applying Eq.~\eqref{eq:parabola} to the combined two-body operator \(H_P+C\) gives
\begin{equation}
    \braket{H_P+C}_{d,z^\star}=\mu_2(d)
    \left[
    H_P(z^\star)+C(z^\star)
    \right].
    \label{eq:combinedparabola}
\end{equation}
The theorem thus delimits what a two-body diagonal catalyst can and cannot change. The shape of the shell-mean curve is fixed by \(\mu_2(d)\), while the catalyst controls its amplitude through \(C(z^\star)\). A catalyst with \(C(z^\star)>0\) flattens or inverts the parabola, whereas one with \(C(z^\star)<0\) pushes the shell means apart.

The latter behavior is visible in Fig.~\ref{fig:catalyst_ordering}(b), where the catalyzed means follow the same parabola with larger amplitude. The shell widths, however, are governed independently by Eq.~\eqref{eq:variance}. A useful catalyst must therefore increase the separation of the shell means relative to their within-shell fluctuations. Both objectives are met by the operator
\(
W^\star_{ij}=-z^\star_i z^\star_j,
\)
for which the energy depends only on \(d_H(z,z^\star)\) and the within-shell variance vanishes. This operator is inaccessible, however, because it requires knowledge of \(z^\star\).

We therefore estimate $W^\star$ one path at a time from the coupling signs, strengths, and graph structure. On the graph of nonzero couplings, shown in Fig. \ref{fig:construction}(a), consider an $m$-edge self-avoiding path  $p=(v_0,\ldots ,v_m)$ and propagate along it the pattern
\begin{equation}
    \label{eq:signprop}
    g^{(p)}_{v_0}=1,\qquad
    g^{(p)}_{v_{a+1}}=-\,\mathrm{sgn}\!\left(J_{v_av_{a+1}}\right)g^{(p)}_{v_a}
\end{equation}
with $a\in[0,m-1]$, which satisfies each bond in turn. Since a path contains no loops, the propagation never contradicts itself, and $g^{(p)}$ is the exact ground state of the path subgraph, computable in time linear in its length. This propagation is illustrated in Fig. \ref{fig:construction}(b). We further assign each vertex a weight given by the magnitude of its incident path couplings: an endpoint takes its single edge $w_{v_0}^{(p)} = |J_{v_0v_1}|$ and $w_{v_m}^{(p)} = |J_{v_{m-1}v_m}|$, while an interior vertex averages its two, $w_{v_i}^{(p)}=\frac{1}{2}(|J_{v_{i-1}v_i}|+|J_{v_iv_{i+1}}|)$, $i\in[1,m-1]$ so that stronger bonds weigh more. The path overlap and its penalty are then
\begin{equation}
    H_p(z)=-M_p(z)^2, \qquad M_p(z) = \sum_{a=0}^{m} w_{v_a}^{(p)} g_{v_a}^{(p)} z_{v_a}.
\end{equation}
minimized when $z$ coincides with $g^{(p)}$ or its global flip along the path. Squaring the overlap $M_p$ preserves the $\mathbb{Z}_2$ symmetry of the cost and, upon expansion, generates pure $ZZ$ couplings with no local fields.  We omit path-edge contributions, which empirically improves funneling.
Let $\mathcal{P}_m$ denote the set of $m$-edge self-avoiding paths, and let \( E(p)=\{(v_a,v_{a+1})\}_{a=0}^{m-1}\) be the edge set of $p$. After discarding the constant $-\sum_{a=0}^{m}(w_{v_a}^{(p)})^2$, summing over all $p\in\mathcal{P}_m$ gives
\begin{equation}
\widetilde{C}_{ij}^{(m)} = -2
\sum_{\substack{
p\in\mathcal{P}_m:\; i,j\in p\\
(i,j)\notin E(p)}}
w_i^{(p)}w_j^{(p)} g_i^{(p)}g_j^{(p)} .
\label{eq:catalyst}
\end{equation}
Thus, each pair $(i,j)$ receives a contribution from all paths containing both vertices, except from paths in which that pair forms an edge. Fig.~\ref{fig:construction}(b) illustrates this for a single path. The final catalyst, $C^{(m)}$, is obtained by projecting $\tilde{C}^{(m)}$ orthogonally to $J$ and normalizing it to unit peak absolute coupling. Paths that agree reinforce one another, and paths that frustration forces to disagree cancel, see Fig.~\ref{fig:construction}(c). The assembled operator is therefore typically denser than $H_P$ and reaches well beyond the original bonds. Before edge removal and projection, each path contributes a weighted rank-one estimate of $W^\star$. The weights, in turn, serve the finer goal. Configurations that frustrate the stronger bonds are penalized more, so $C^{(m)}$ tends to order states within a shell as $H_P$ does. For $m=2$, the catalyst is proportional to the off-diagonal part of $-J^2$. Raising the order $m$ extends the reach of the catalyst but degrades the fidelity of its patterns to the true ground state, a trade-off we quantify in Appendix~\ref{app:order}.\

\begin{figure}[t]
\centering
\includegraphics[width=\linewidth]{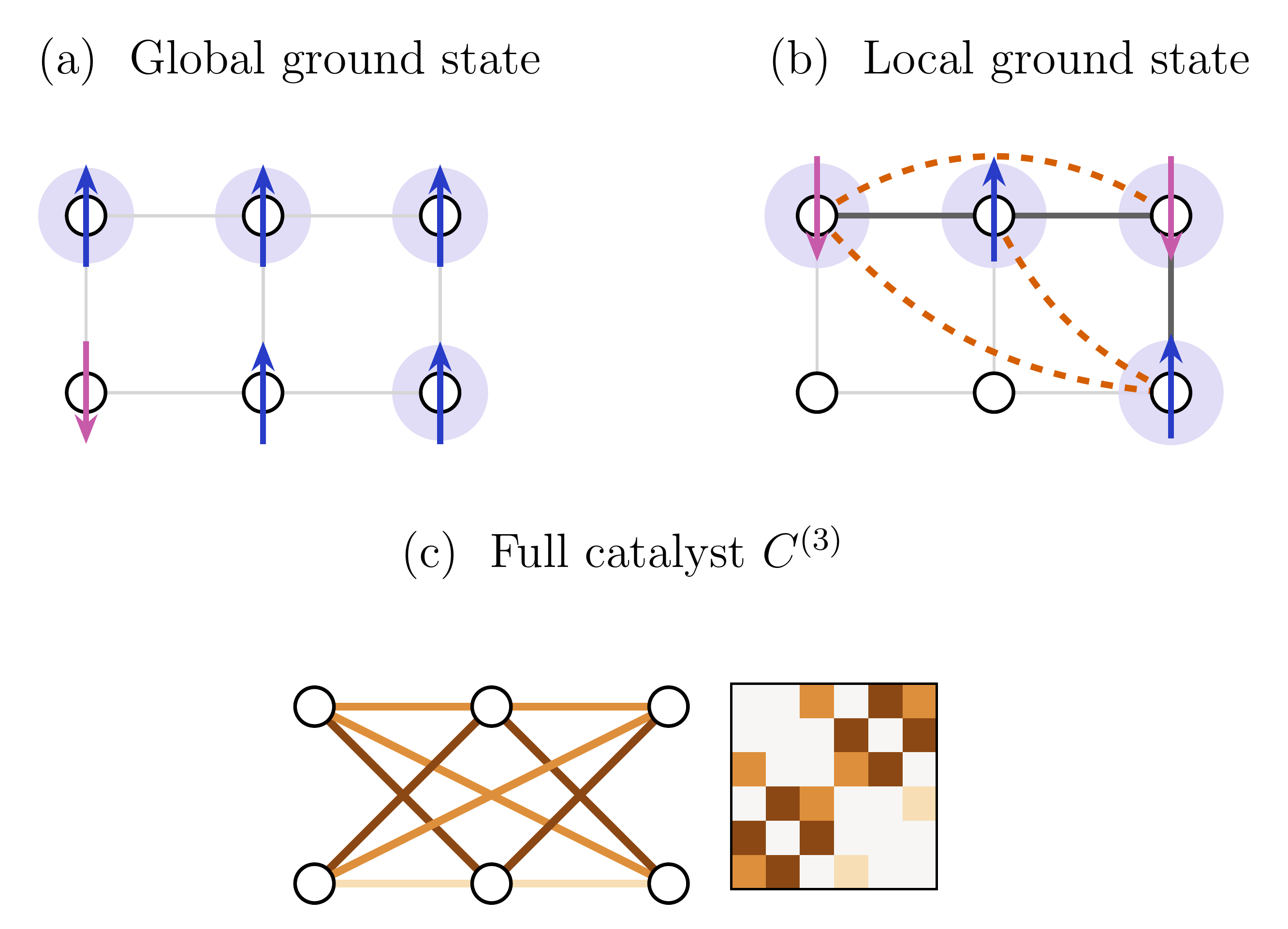}
\caption{\textbf{Open-path catalyst construction.}
(a) A global ground-state pattern \(z^\star\) minimizes the full signed problem but is unknown to the construction. (b) On a selected open path \(p\), sign propagation from Eq.~\eqref{eq:signprop} gives a locally satisfying pattern \(g^{(p)}\), which acts as a solution-independent estimate of the ground-state alignment on that path. The path penalty induces couplings between nonadjacent path vertices. (c) Summing all order-three path contributions yields the denser catalyst \(C^{(3)}\) after orthogonalization and normalization.}
\label{fig:construction}
\end{figure}

\emph{Results}---We benchmark the construction on \(200\) random \(3\)-regular instances with \(n=20\) and couplings drawn independently and uniformly from \([-1,1]\). At this size, the full classical spectrum and exact ground-state set \(\mathcal G\) are obtained by enumerating all \(2^n\) configurations. The set \(\mathcal G\) is used only for diagnostics, through Hamming distances and final-state binning, whereas the catalyst is constructed from \(J\) alone. We first verify that the catalyst reshapes the static landscape as predicted by the shell analysis (Fig.~\ref{fig:summary}, Appendix~\ref{app:static-analysis}).

We then simulate the same instances at sweep times \(T=2,6,10\). The time-dependent Hamiltonian is
\begin{equation}
H(s)=-(1-s)\sum_iX_i+sH_P+s(1-s)C^{(m)},
\end{equation}
so the final Hamiltonian and ground-state manifold are unchanged. Unless stated otherwise, quantitative comparisons use \(m=4\), while Fig.~\ref{fig:results} additionally displays orders \(2\)--\(5\).

Fig.~\ref{fig:results} shows the redistribution of population relative to the uncatalyzed protocol. In energy space, configurations are binned by the normalized energy $E/E_{\mathrm{GS}}$, where $E_{\mathrm{GS}}<0$ is the ground-state energy, so that ground states satisfy $E/E_{\mathrm{GS}}=1$. The catalyst shifts probability away from intermediate-quality energy bands and toward the best band near $E/E_{\mathrm{GS}}=1$.

We quantify the shift by the top-band mass $P_{0.95}\equiv P\!\left(E/E_{\mathrm{GS}}\ge 0.95\right)$ and, for each instance, the paired gain $\bigl(P_{0.95}^{\mathrm{cat}}-P_{0.95}^{\mathrm{van}}\bigr)/P_{0.95}^{\mathrm{van}}$. At $T=6$ the catalyst raises the median top-band mass from $0.067$ to $0.324$, with a median paired gain of $+351\%$ (interquartile range $+233\%$ to $+503\%$) and $99\%$ of instances improved; at $T=10$ the median absolute gain reaches $+0.36$ with $93\%$ of instances improved. Results are qualitatively unchanged for band thresholds $0.90$ and $0.98$ \cite{SM}.

The same redistribution is visible in Hamming space as a shift toward the ground-state manifold. The catalyzed anneal increases the probability mass at small $\delta$, most visibly in the near-solution region $\delta\le2$, where at $T=6$ the median mass rises from $0.10$ to $0.35$ ($+239\%$, $94\%$ of instances), while draining it from intermediate distances, as shown in the bottom row of Fig.~\ref{fig:results}. This concentration is valuable even when the ground state itself is not sampled, since nearby configurations can be refined with local post-processing.

Throughout, we normalize every catalyst to unit maximum coupling. The path construction can exploit additional available couplers and thereby increase the total interaction strength; this is an intended feature of the protocol. To check that the benefit comes from the added structure and not merely from this extra energy, we compare against the reinforced-problem baseline $C = H_P$, which raises the energy scale by rescaling the original couplings without activating any new coupling directions. At the shortest sweep time, $T = 2$, this reinforcement already improves the output distribution. At longer sweep times the open-path catalysts clearly outperform this equal peak-coupling control: at $T=6$ the median paired improvement over $C=H_P$ is $+77\%$ in top-band mass ($92\%$ of instances) and $+102\%$ in near-solution Hamming mass. Thus, distributing the available coupling strength according to the path construction is more effective than uniformly reinforcing $H_P$. This is consistent with the reduction of the scale-invariant shell-overlap metric in Fig.~\ref{fig:summary}(c) of Appendix~\ref{app:static-analysis}, which is independent of the energy scale by construction.

\begin{figure*}[t]
\centering
\includegraphics[width=\linewidth]{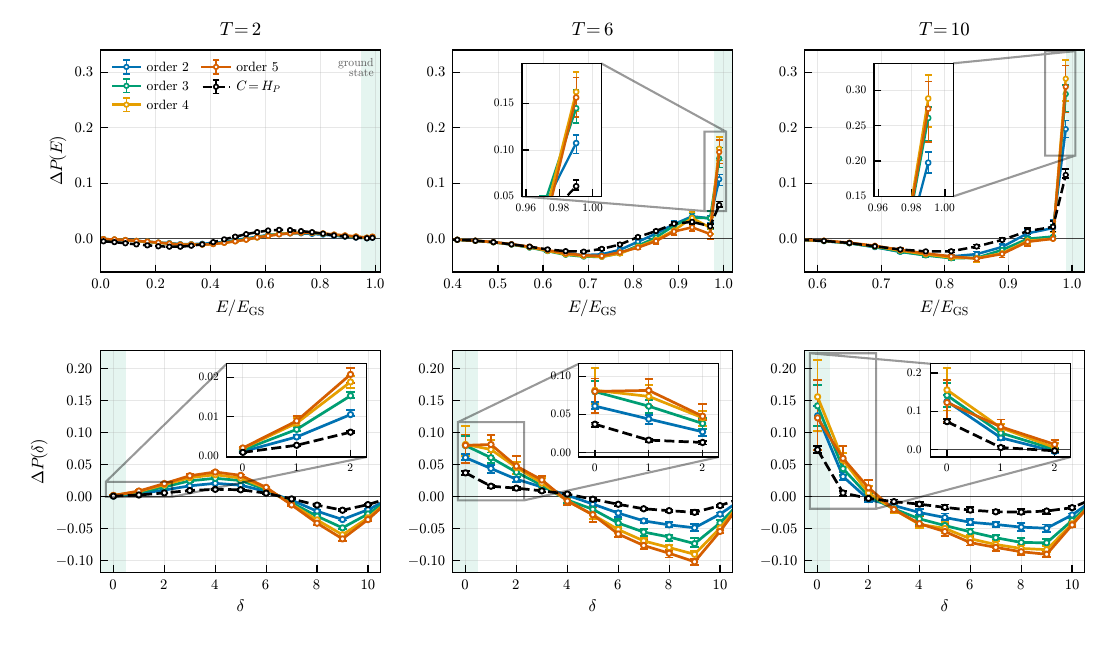}
\caption{\textbf{Probability redistribution.} Final-state probability redistribution on \(200\) \(3\)-regular instances with \(n=20\). Results are shown for sweep times \(T=2,6,10\), for open-path catalysts of orders \(2\)--\(5\), and for the reinforced-problem baseline \(C=H_P\), all relative to the uncatalyzed anneal. Curves show the median change in probability mass across instances; error bars are \(95\%\) bootstrap confidence intervals with \(B=5000\) resamples. Top: redistribution by normalized energy \(E/E_{\mathrm{GS}}\), with \(E/E_{\mathrm{GS}}=1\) corresponding to the ground-state energy. Bottom: redistribution by Hamming distance \(\delta\) to the closest ground state.}
\label{fig:results}
\end{figure*}

We repeated the analysis on denser graph families to assess how the construction depends on graph connectivity. The gains persist with an ordered decay in connectivity: at $T=6$ the median top-band improvement over the uncatalyzed anneal is $+351\%$, $+184\%$, and $+133\%$ for \(3\)-, \(4\)-, and \(5\)-regular graphs respectively, with all instances improved at $T=2$ on every family. The order that maximizes the median paired gain in top-band mass at the longer sweep times ($T=6,10$) decreases with degree ($m=4,3,2$), consistent with the trade-off of Appendix~\ref{app:order}. On the fully connected Sherrington--Kirkpatrick model the energetic gains are reduced but remain clear at short and long sweeps (median $+22\%$ and $+16\%$ in $P_{0.95}$ at $T=2$ and $T=10$), while the Hamming funneling persists at all sweep times ($+56\%$ median in $P(\delta\le2)$ at $T=2$, with every instance improved).

The choice of subgraph family is thus a design axis that sets where the catalyst places its emphasis. Richer subgraphs resolve more local structure and so sharpen the funnel in Hamming distance, but the more relative orientations a pattern fixes, the more the catalyst overwrites the ordering of the problem within each shell. For paths specifically, in dense graphs a long path ignores many chords of its induced subgraph, so the sign-propagated pattern of Eq.~\eqref{eq:signprop} becomes a less faithful ground-state estimate as the order grows. Open paths at moderate order offer a practical compromise.
\newline\newline
\emph{Discussion and conclusion}---We introduced two-body diagonal catalysts for quantum annealing that are constructed solely from the sign and magnitude structure of the problem couplings, without using information about the solution. The catalysts target the Hamming-distance geometry by deepening the mean shell funnel and reducing the overlap between Hamming shells, thereby shifting the final annealing distribution toward the ground-state manifold and its immediate neighborhood. More broadly, the shell-moment theorem provides static criteria for
catalyst design, allowing candidate constructions to be screened through
funnel depth and shell overlap before any annealing dynamics study. Being two-local and diagonal, they preserve the transverse-field structure and use the control type native to the annealers, subject to connectivity and coupling-range constraints~\cite{Johnson2011}.

Several extensions follow directly. In a digital setting, \(C\) can be inserted into QAOA \cite{Farhi2014} as an additional diagonal phase separator. The catalyst can also be combined with optimized diabatic annealing protocols~\cite{CrossonLidar2021,Feinstein2025,Werner2026}. Reverse annealing is another natural addition because its performance depends strongly on the Hamming distance between the initial state and the optimum~\cite{Ohkuwa2018ReverseAnnealingPSpin}. The scaling with system size also remains to be characterized, as well as generalizing the formalism beyond the field-free \(ZZ\) cost functions. Finally, closed-form average-case guarantees for the funneling coefficients, as a function of the problem and subgraph family, would make catalyst selection a principled decision. Taken together, these extensions could help improve the regimes in which quantum annealing may outperform classical solvers for approximate optimization~\cite{LidarScaling, qcws-8kgm}.
\newline\newline
\emph{Acknowledgments}---The authors acknowledge RES resources provided by Barcelona Supercomputing Center in MareNostrum 5 to INNO-2026-1-0004, and thank our colleagues from the Hackamonth in Singapore.

\bibliography{bibliography}

\appendix
\setcounter{secnumdepth}{3}

\section{Proof of Theorem 1}
\label{app:shell}

\emph{Proof.}---Fix a reference configuration \(z^\star\) and define the agreement variables \(x_i=z_i z_i^\star\in\{\pm1\}\), together with the aligned couplings \(\hat W_{ij}=W_{ij}z_i^\star z_j^\star\). Any two-body operator \(H(z)=\sum_{i<j}W_{ij}z_i z_j\) then reads \(H=\sum_{i<j}\hat W_{ij}x_i x_j\), with \(H(z^\star)=\sum_{i<j}\hat W_{ij}\). In these variables, a configuration in \(S_{d,z^\star}\) is specified by the set \(F=\{i:x_i=-1\}\) of its flipped positions, with \(|F|=d\). The uniform average over the shell is therefore an average over the \(\binom{n}{d}\) possible choices of \(F\).

Consider now the average of a product of spins over \(k\) distinct positions, \(\mu_k(d)=\braket{x_{i_1}\cdots x_{i_k}}_d\), from which every shell statistic below is built. The product depends only on how many of the \(k\) marked positions are flipped, \(x_{i_1}\cdots x_{i_k}=(-1)^r\), with \(r=|F\cap\{i_1,\dots,i_k\}|\). Among the \(\binom{n}{d}\) configurations of the shell, exactly \(\binom{k}{r}\binom{n-k}{d-r}\) contain \(r\) flips on the marked positions. Hence
\begin{equation}
    \label{eq:hypergeom}
    \mu_k(d) = \binom{n}{d}^{-1} \sum_{r=0}^{k} (-1)^r \binom{k}{r} \binom{n-k}{d-r},
\end{equation}
with the convention \(\binom{a}{b}=0\) for \(b<0\) or \(b>a\). Permutation symmetry of the shell makes the result independent of the chosen positions, while Vandermonde's identity \(\sum_r\binom{k}{r}\binom{n-k}{d-r}=\binom{n}{d}\) confirms the normalization.

For \(k=2\), the three terms of Eq.~\eqref{eq:hypergeom} sum to \((n-d)(n-d-1)-2d(n-d)+d(d-1)\) over \(n(n-1)\), giving
\begin{equation}
    \mu_2(d) = \frac{(n-2d)^2-n}{n(n-1)}.
\end{equation}
The shell mean then follows by linearity,
\begin{equation}
    \langle H\rangle_{d,z^\star} =
\sum_{i<j} \hat W_{ij} \langle x_i x_j\rangle_d
= \mu_2(d)H(z^\star),
\end{equation}
which proves Eq.~\eqref{eq:parabola}.

For the variance, write \(H=\sum_P\hat W_Px_P\), with \(P=\{i,j\}\) and \(x_P=x_i x_j\), so that
\begin{equation}
    \langle H^2\rangle_{d,z^\star} = \sum_{P,Q} \hat W_P\hat W_Q \langle x_Px_Q\rangle_{d,z^\star}.
\end{equation}

Since \(x_i^2=1\), the ordered pairs \((P,Q)\) fall into three classes. If \(P=Q\), the product equals \(1\) and the total weight is \(S_2=\sum_{i<j}\hat W_{ij}^2=\sum_{i<j}W_{ij}^2\). If \(P\) and \(Q\) share one index, two distinct spins remain and average to \(\mu_2\), with total weight \(U\). Expanding the squared row sums gives
\begin{equation}
    \sum_i\left(\sum_j\hat W_{ij}\right)^2=2S_2+U.
\end{equation}
Since \(\sum_j\hat W_{ij}=z_i^\star(Wz^\star)_i\), this sum also equals \(\lVert Wz^\star\rVert^2\), and therefore \(U=\lVert Wz^\star\rVert^2-2S_2\). If \(P\) and \(Q\) are disjoint, four distinct spins remain and average to \(\mu_4\), with total weight \(V\). The three classes exhaust the square of the total coefficient sum, \(H(z^\star)^2=S_2+U+V\). Hence
\begin{equation}
    \langle H^2\rangle_{d,z^\star} = S_2+\mu_2U + \mu_4\bigl[H(z^\star)^2-S_2-U\bigr],
\end{equation}
and subtracting \(\langle H\rangle_{d,z^\star}^2=\mu_2^2H(z^\star)^2\) proves Eq.~\eqref{eq:variance}.

For \(n\geq4\), evaluating Eq.~\eqref{eq:hypergeom} at \(k=4\) gives
\begin{equation}
    \label{eq:moment_4}
    \mu_4(d) = \frac{(n-2d)^4-(6n-8)(n-2d)^2+3n^2-6n}{n(n-1)(n-2)(n-3)},
\end{equation}
which proves Eq.~\eqref{eq:fourmoment}. Like \(\mu_2\), it equals unity at \(d=0,n\) and approaches the independent-flip value \([(n-2d)/n]^4\) at large \(n\).

Both identities hold for any reference \(z^\star\), not only a ground state. For the ideal couplings \(W^\star_{ij}=-z_i^\star z_j^\star\), \(i<j\), the cost depends only on \(d\) and \(\operatorname{Var}_{d,z^\star}(H^\star)=0\): every shell collapses onto its mean, which is the limit approximated by our construction. \hfill\(\square\)
\section{Order--fidelity trade-off}
\label{app:order}

Let \(\phi_p\in[0,1]\) be the fraction of vertices of path \(p\) on which the propagated pattern agrees with the reference solution, \(g^{(p)}_{v_a}=z^\star_{v_a}\). For the unprojected construction with unit weights and without path-edge removal, an order-\(m\) path has \(m+1\) vertices and
\begin{equation}
    M_p(z^\star)=(m+1)(2\phi_p-1).
    \label{eq:Mp_solution_fidelity}
\end{equation}
Expanding \(-M_p^2\), discarding the constant term, and summing over paths gives
\begin{equation}
    C_{\mathrm{id}}(z^\star) = -\sum_{p\in\mathcal P_m} (m+1)
    \left[
        (m+1)(2\phi_p-1)^2-1
    \right].
    \label{eq:funnelcoeff}
\end{equation}
An individual path therefore deepens the shell-mean funnel when
\begin{equation}
    (2\phi_p-1)^2>\frac{1}{m+1}.
    \label{eq:fidelity_threshold}
\end{equation}
Because \(M_p^2\) is invariant under a global flip of the propagated pattern, fidelities \(\phi_p\) and \(1-\phi_p\) are equivalent, and only \(\lvert2\phi_p-1\rvert\) matters. For independent random agreement, \(\mathbb{E}[(2\phi_p-1)^2]=1/(m+1)\), so Eq.~\eqref{eq:fidelity_threshold} requires a squared alignment above the random baseline. Increasing \(m\) extends the reach of the construction but can reduce the fidelity of the propagated patterns, since the exact path ground state ignores chords and loops that also constrain \(z^\star\). The preferred order is therefore instance dependent and may be larger for lower-degree or less-frustrated graphs.
\section{Structural properties}
\label{app:static-analysis}

We test whether the open-path catalyst reshapes the classical landscape as predicted by the shell analysis. Fig.~\ref{fig:summary} collects static diagnostics for the \(200\) signed \(3\)-regular instances with \(n=20\) introduced in the main text. Their exact ground-state sets \(\mathcal G\) are obtained by enumeration of all \(2^n\) configurations. Shell means and variances are evaluated using Eqs.~\eqref{eq:parabola} and \eqref{eq:variance}, which we verified against direct shell enumeration.

\begin{figure}[!h]
\centering
\includegraphics[width=0.8\linewidth]{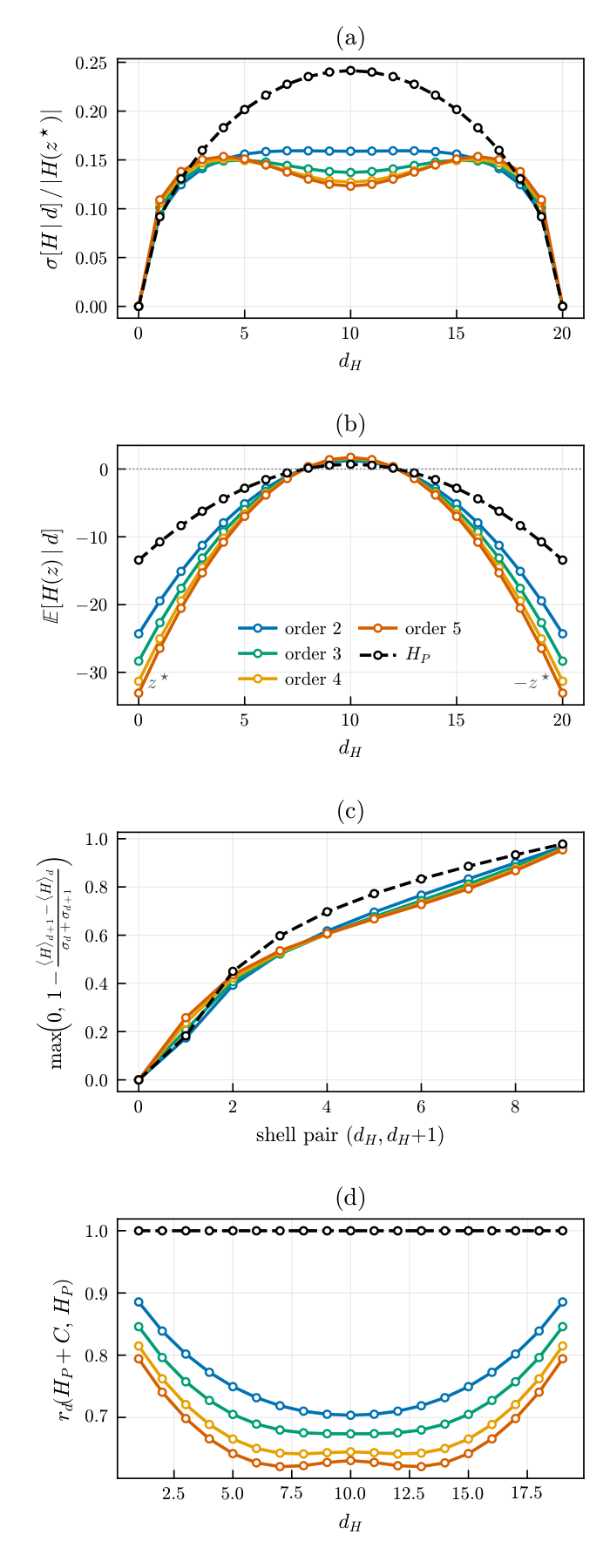}
\caption{\textbf{Static landscape diagnostics.}
Results on \(200\) signed \(3\)-regular instances with \(n=20\). The catalyzed operators \(H_{\mathrm{comb}}^{(m)}=(H_P+C^{(m)})/\kappa_m\), with \(m=2,\ldots,5\), are compared with \(H_P\). (a) Relative within-shell spread \(\sigma[H\mid d]/|H(z^\star)|\). (b) Shell mean \(\langle H\rangle_d\). (c) Adjacent-shell overlap \(\mathrm{ov}_d\). (d) Within-shell Pearson correlation \(r_d(H_{\mathrm{comb}}^{(m)},H_P)\). The catalyst deepens the shell-mean funnel, narrows the shells relatively, and lowers their scale-invariant overlap while retaining positive correlation with the original within-shell ordering.}
\label{fig:summary}
\end{figure}

Each combined operator is normalized to the same peak coupling as the bare problem,
\begin{equation}
    H_{\mathrm{comb}}^{(m)} = \frac{H_P+C^{(m)}}{\kappa_m}, \qquad \kappa_m =
    \max_{i<j}\left|J_{ij}+C^{(m)}_{ij}\right|.
    \label{eq:static_combined_operator}
\end{equation}
The catalyst acts on both fronts identified in the main text. The shell means separate, with \(\lvert H_{\mathrm{comb}}^{(m)}(z^\star)\rvert\) increasing from approximately \(13\) for \(H_P\) to \(24\)--\(33\), depending on the order Fig.~\ref{fig:summary}(b). At the same time, the typical relative shell width \(\sigma_d/|H(z^\star)|\) decreases from about \(0.19\) to \(0.14\) Fig.~\ref{fig:summary}(a).

We quantify the resulting separation of neighboring shells on the near-solution branch \(d<n/2\) through
\begin{equation}
    \mathrm{ov}_d = \max\left[0,\, 1-
    \frac{
    \langle H\rangle_{d+1}-\langle H\rangle_d
    }{
    \sigma_d+\sigma_{d+1}
    }
    \right].
    \label{eq:shell_overlap}
\end{equation}
This measures the overlap of adjacent one-standard-deviation bands: \(\mathrm{ov}_d=1\) corresponds to coincident bands and \(\mathrm{ov}_d=0\) to separated bands. The overlap decreases from approximately \(0.70\) for \(H_P\) to \(0.65\) for the catalyzed operators Fig.~\ref{fig:summary}(c). Although modest, this reduction is systematic and cannot arise from a global rescaling, since \(\mathrm{ov}_d\) is invariant under \(H\to cH\). The catalyst therefore changes the ratio of shell spacings to shell widths, rather than only the overall energy scale.

To assess whether this reshaping preserves the original energetic ordering within each shell, we compute the shell-resolved Pearson correlation \(r_d(H_{\mathrm{comb}}^{(m)},H_P)\) over configurations in \(S_{d,z^\star}\). It remains positive and sizeable, particularly near the solution Fig.~\ref{fig:summary}(d), showing that the catalyst improves inter-shell separation without fully scrambling the ordering within each shell. The static gains saturate with order, with \(m=3\)--\(5\) giving similar diagnostics on these \(3\)-regular instances, consistent with Appendix~\ref{app:order}.

Finally, peak-coupling normalization does not fix the Frobenius norm, because the path construction introduces additional two-body couplings. At the orders used in the main comparisons, however, the median ratio \(\lVert C^{(m)}\rVert_F/\lVert H_P\rVert_F\) is \(1.08\) (\(m=4\), \(3\)-regular), \(1.01\) (\(m=3\), \(4\)-regular), and \(0.97\) (\(m=2\), \(5\)-regular), and lies below unity on \(28\%\), \(48\%\), and \(59\%\) of instances, respectively Fig.~\ref{fig:sm-frobenius} of the Supplemental Material. Thus, at the selected orders the catalysts have total coupling weight comparable to that of \(H_P\), while the reduction of \(\mathrm{ov}_d\) is independent of energy scale by construction.

\clearpage
\onecolumngrid
\begin{center}
  \textbf{\large Supplemental Material for\\[2pt]
  ``Reshaping quantum annealing landscapes
  with diagonal catalysts''}
\end{center}
\suppressfloats[t]
\setcounter{figure}{0}   \renewcommand{\thefigure}{S\arabic{figure}}
\setcounter{table}{0}    \renewcommand{\thetable}{S\arabic{table}}
\setcounter{equation}{0} \renewcommand{\theequation}{S\arabic{equation}}
\section*{Relative-improvement statistics}

This Supplemental Material documents the paired per-instance statistics quoted in the main text and shows their robustness to the choice of energy threshold, graph family, and baseline. For each graph family of three-, four-, and five-regular graphs with couplings drawn independently and uniformly from $[-1,1]$, with $200$ instances each, and $n=20$, every protocol is simulated on the same instances, so each protocol is compared with the baseline instance by instance.\\

Our comparisons rest on two observables of the final-state distribution. In energy space, the top-band mass $P_{\rho_0} \equiv P\!\left(E/E_{\mathrm{GS}} \ge \rho_0\right)$ is the probability of sampling within a fraction $\rho_0$ of the ground-state energy. The main text quotes $\rho_0=0.95$, and we also report $\rho_0=0.90$ and $0.98$. In Hamming space, the near-solution mass $P(\delta\le2)$ is the probability of sampling within two spin flips of the ground-state manifold. The improvements are always assessed instance by instance. For an observable $P$, the paired gain over a baseline is $\left(P^{\mathrm{cat}} - P^{\mathrm{base}}\right)/P^{\mathrm{base}}$, with both values evaluated on the same instance, and we report the median and interquartile range of this gain across the ensemble, together with the fraction of instances that strictly improve. The path order is fixed per family, $m=4,3,2$ for the three-, four-, and five-regular graphs, with no per-panel or per-instance tuning.\\

We first compare against the uncatalyzed anneal. Fig.~\ref{fig:sm-rho95} shows the catalyzed versus uncatalyzed top-band mass $P_{0.95}$, one point per instance, for the three families and sweep times $T=2,6,10$ on logarithmic axes. Points above the diagonal are improved instances, and each panel is annotated with the median paired gain and the fraction improved. Three features stand out. First, the improvement is near-uniform across the ensemble. At $T=2$ and $T=6$ essentially every instance improves, and at $T=10$ the degraded instances form a small tail, between $5$ and $15\%$ depending on the family. Second, relative gains decrease with sweep time as the uncatalyzed anneal becomes more adiabatic, while absolute gains grow. Third, gains decrease smoothly with connectivity, consistent with the order--fidelity trade-off of Appendix~B of the main text. Figs.~\ref{fig:sm-rho90} and \ref{fig:sm-rho98} repeat the analysis at $\rho_0=0.90$ and $0.98$. Tightening the band increases the relative gains and mildly reduces the fraction improved, leaving the qualitative picture unchanged. The same conclusions hold in Hamming space, where Fig.~\ref{fig:sm-hamming-van} uses the near-solution mass $P(\delta\le2)$ in place of the energy band. Table~\ref{tab:sm-vanilla} collects the medians.\\ 

We next compare against the reinforced-problem baseline. Figs.~\ref{fig:sm-baseline} and \ref{fig:sm-hamming-base} show the same per-instance comparisons with the equal peak-coupling control $C=H_P$ in place of the bare anneal. The gains remain positive for the large majority of instances, showing that uniform reinforcement of the original problem couplings does not reproduce the benefit of the open-path construction. Table~\ref{tab:sm-baseline} collects the corresponding medians.\\

Finally, we test the construction on fully connected instances. Figs.~\ref{fig:sm-sk-energy} and \ref{fig:sm-sk-hamming} show the paired comparisons for $200$ instances of Sherrington--Kirkpatrick type, with couplings drawn independently and uniformly from $[-1,1]$, for all path orders $m=2$--$5$. On the complete graph no order is singled out, and the results are nearly order-independent. In energy space the gains are strongly reduced relative to the sparse families but remain clear at short and long sweeps. For $m=2$, the median paired gain in $P_{0.95}$ is $+22\%$ at $T=2$, with $79\%$ of instances improved, and $+16\%$ at $T=10$, with $88\%$ improved. The dependence on sweep time is non-monotone, and near $T=6$ the median gain over the uncatalyzed anneal crosses zero, $+3\%$ in $P_{0.95}$ and $+8\%$ at the tightest band $\rho_0=0.98$. Even there, however, the catalyst outperforms the reinforced-problem control by $+20\%$, because the reinforcement alone degrades the uncatalyzed output, with a median $P_{0.95}$ of $0.200$ against $0.226$. The Hamming concentration, in contrast, is robust at all sweep times. The median paired gain in $P(\delta\le2)$ is $+56\%$ at $T=2$, with every instance improved, $+28\%$ at $T=6$, and $+25\%$ at $T=10$. These results realize the dense limit of the order--fidelity trade-off of Appendix~\ref{app:order}. On the complete graph the propagated path patterns retain little fidelity to the energetic ordering of the problem, which compresses the energy-band gains from severalfold to tens of percent, while the geometric concentration toward the ground-state manifold survives intact.

\begin{figure*}[p]
\centering
\includegraphics[width=\linewidth]{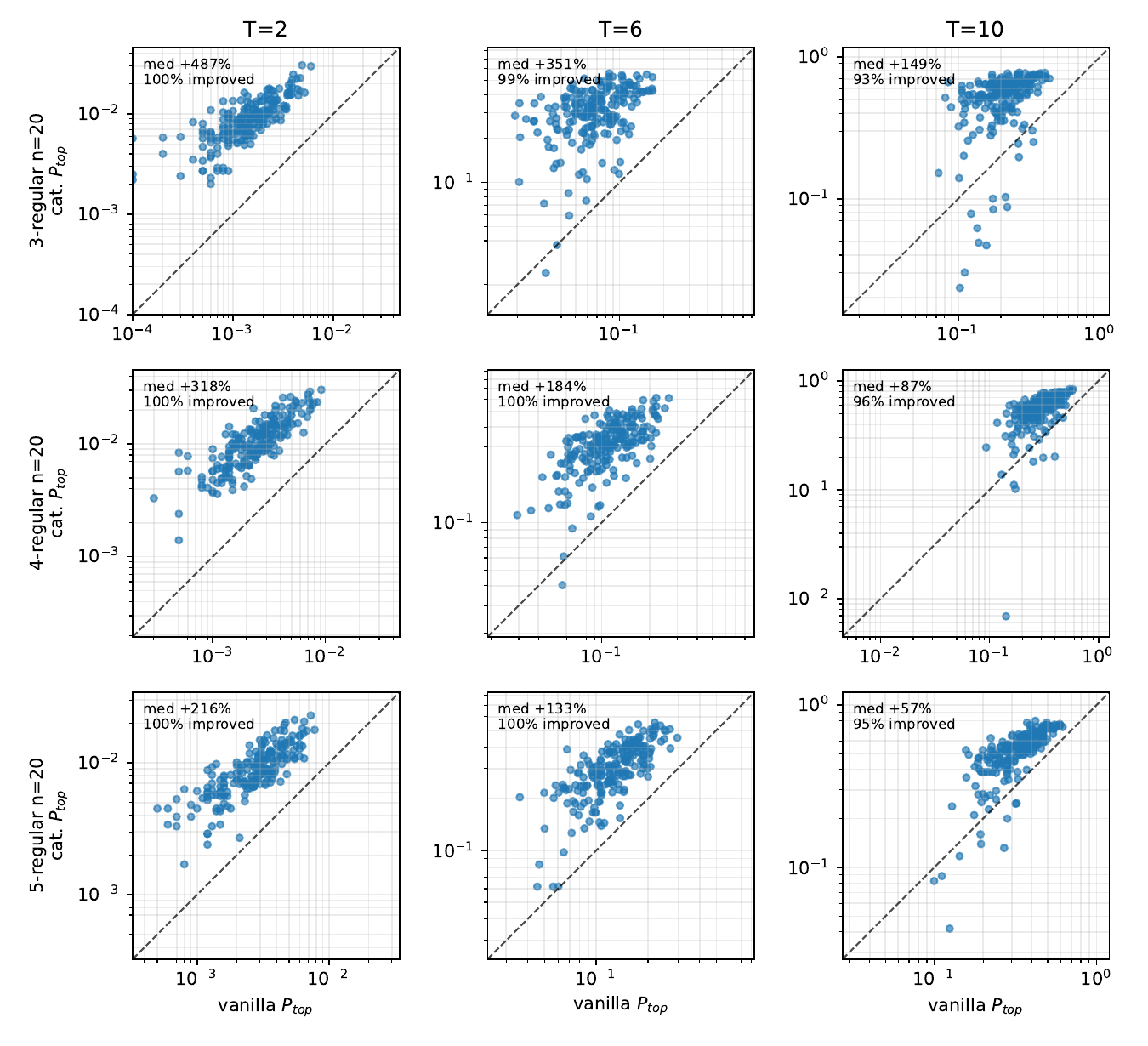}
\caption{\textbf{Paired top-band improvement over the uncatalyzed anneal.} Catalyzed versus uncatalyzed top-band mass $P_{0.95}$, one point per instance, on logarithmic axes. Rows: 3-, 4-, and 5-regular families ($200$ instances each, $n=20$, fixed orders $m=4,3,2$). Columns: sweep times $T=2,6,10$. The dashed line marks equality; annotations give the median paired gain and the fraction of instances improved.}
\label{fig:sm-rho95}
\end{figure*}

\begin{figure*}[p]
\centering
\includegraphics[width=\linewidth]{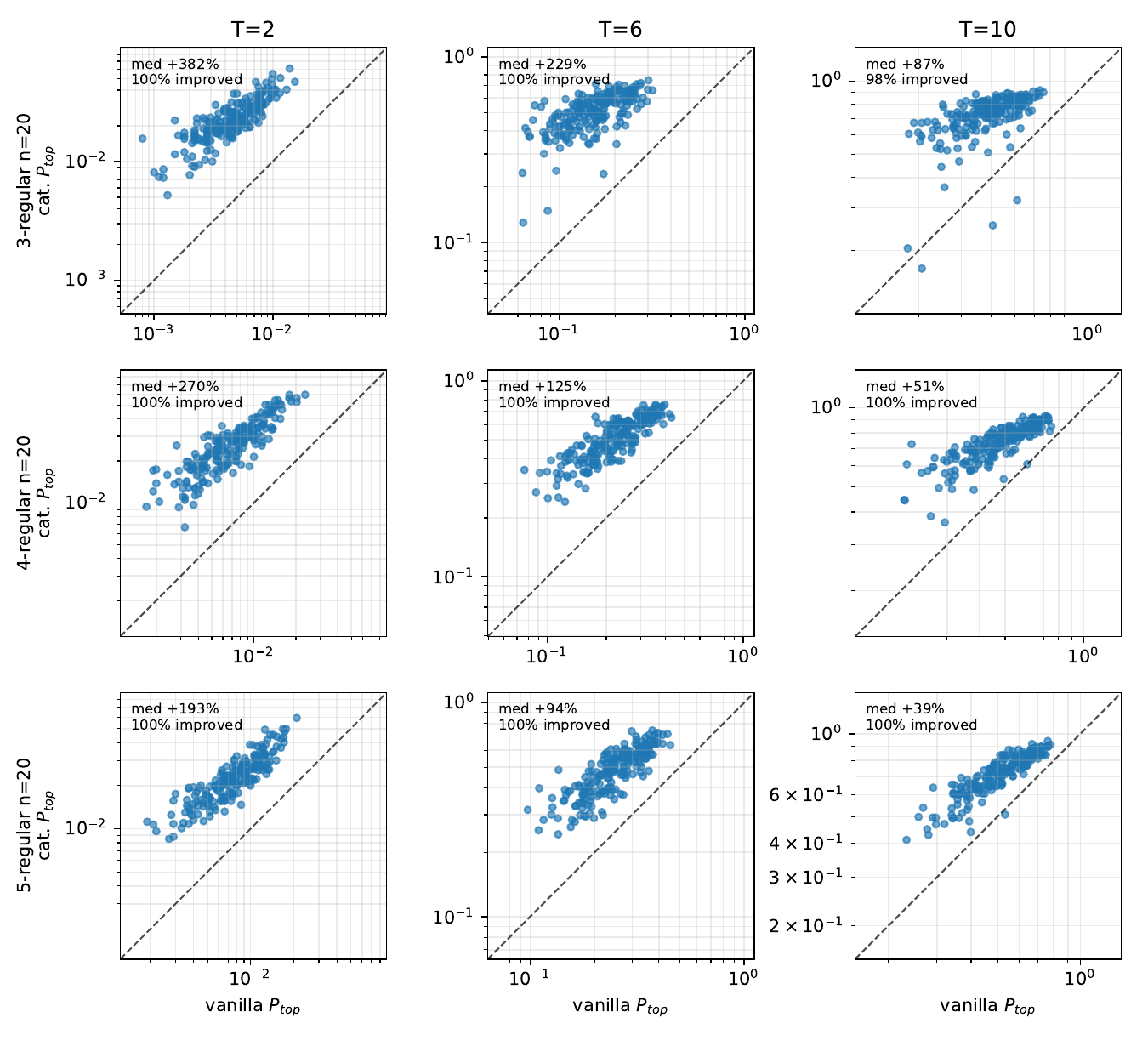}
\caption{\textbf{Threshold robustness, $\rho_0=0.90$.} As in Fig.~\ref{fig:sm-rho95}, for band threshold $\rho_0=0.90$.}
\label{fig:sm-rho90}
\end{figure*}

\begin{figure*}[p]
\centering
\includegraphics[width=\linewidth]{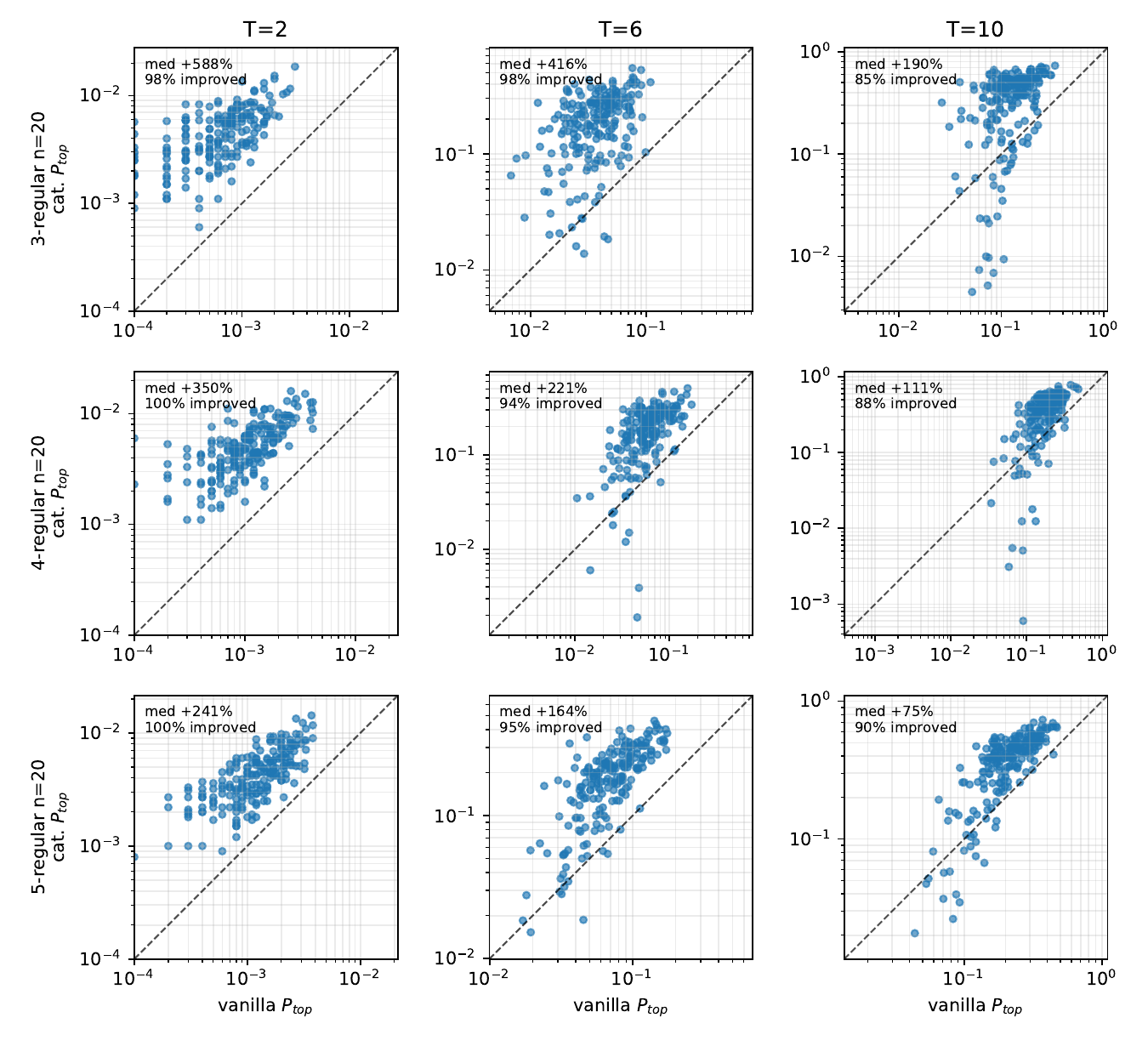}
\caption{\textbf{Threshold robustness, $\rho_0=0.98$.} As in Fig.~\ref{fig:sm-rho95}, for band threshold $\rho_0=0.98$.}
\label{fig:sm-rho98}
\end{figure*}

\begin{figure*}[p]
\centering
\includegraphics[width=\linewidth]{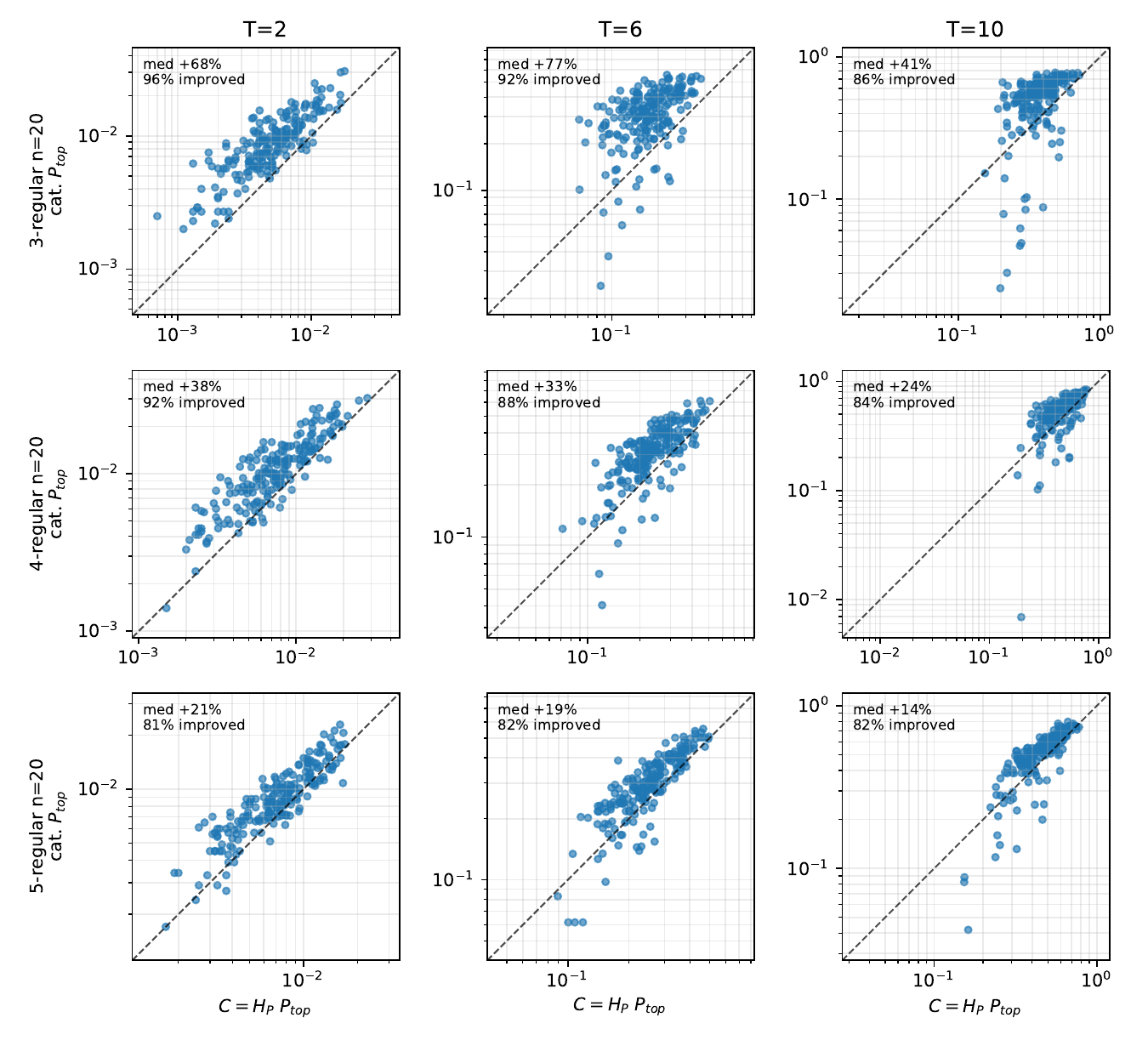}
\caption{\textbf{Paired top-band improvement over the reinforced-problem baseline.} As in Fig.~\ref{fig:sm-rho95}, with the equal-budget control $C=H_P$ as the baseline instead of the uncatalyzed anneal.}
\label{fig:sm-baseline}
\end{figure*}

\begin{figure*}[p]
\centering
\includegraphics[width=\linewidth]{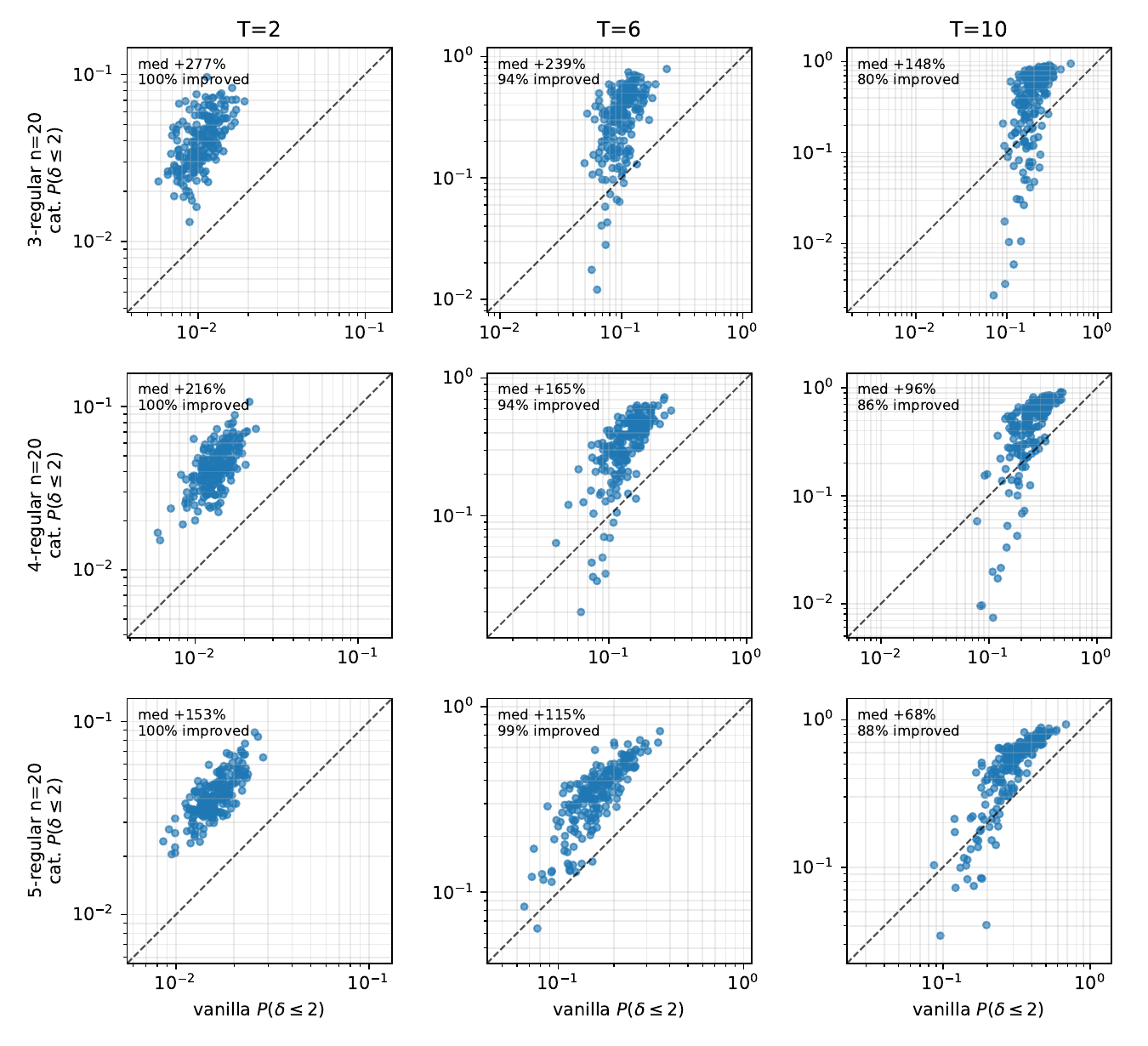}
\caption{\textbf{Near-solution Hamming mass, uncatalyzed baseline.} Catalyzed versus uncatalyzed $P(\delta\le2)$, one point per instance.}
\label{fig:sm-hamming-van}
\end{figure*}

\begin{figure*}[p]
\centering
\includegraphics[width=\linewidth]{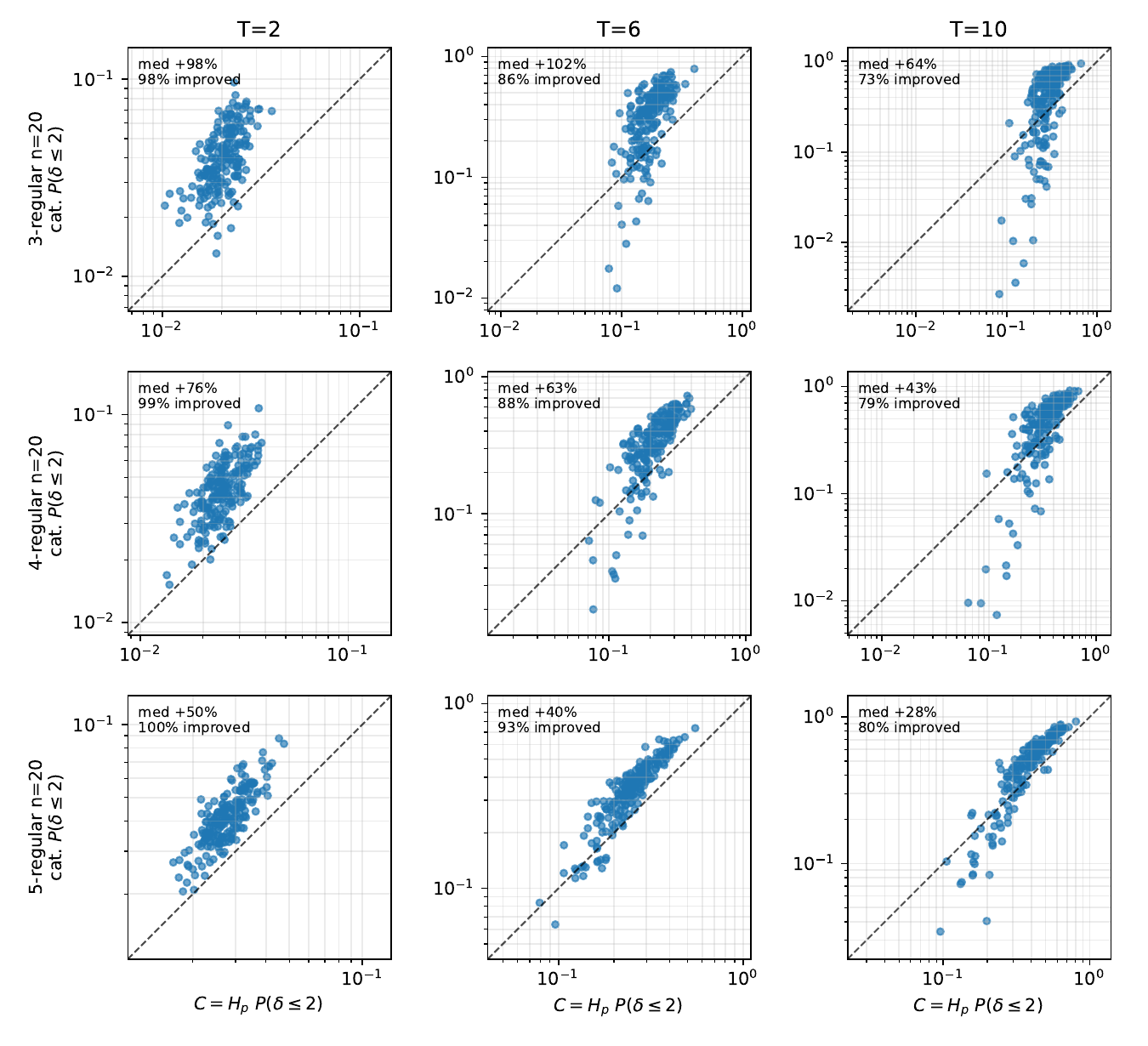}
\caption{\textbf{Near-solution Hamming mass, reinforced-problem baseline.} Catalyzed versus $C=H_P$ $P(\delta\le2)$, one point per instance.}
\label{fig:sm-hamming-base}
\end{figure*}

\begin{figure*}[p]
\centering
\includegraphics[width=\linewidth]{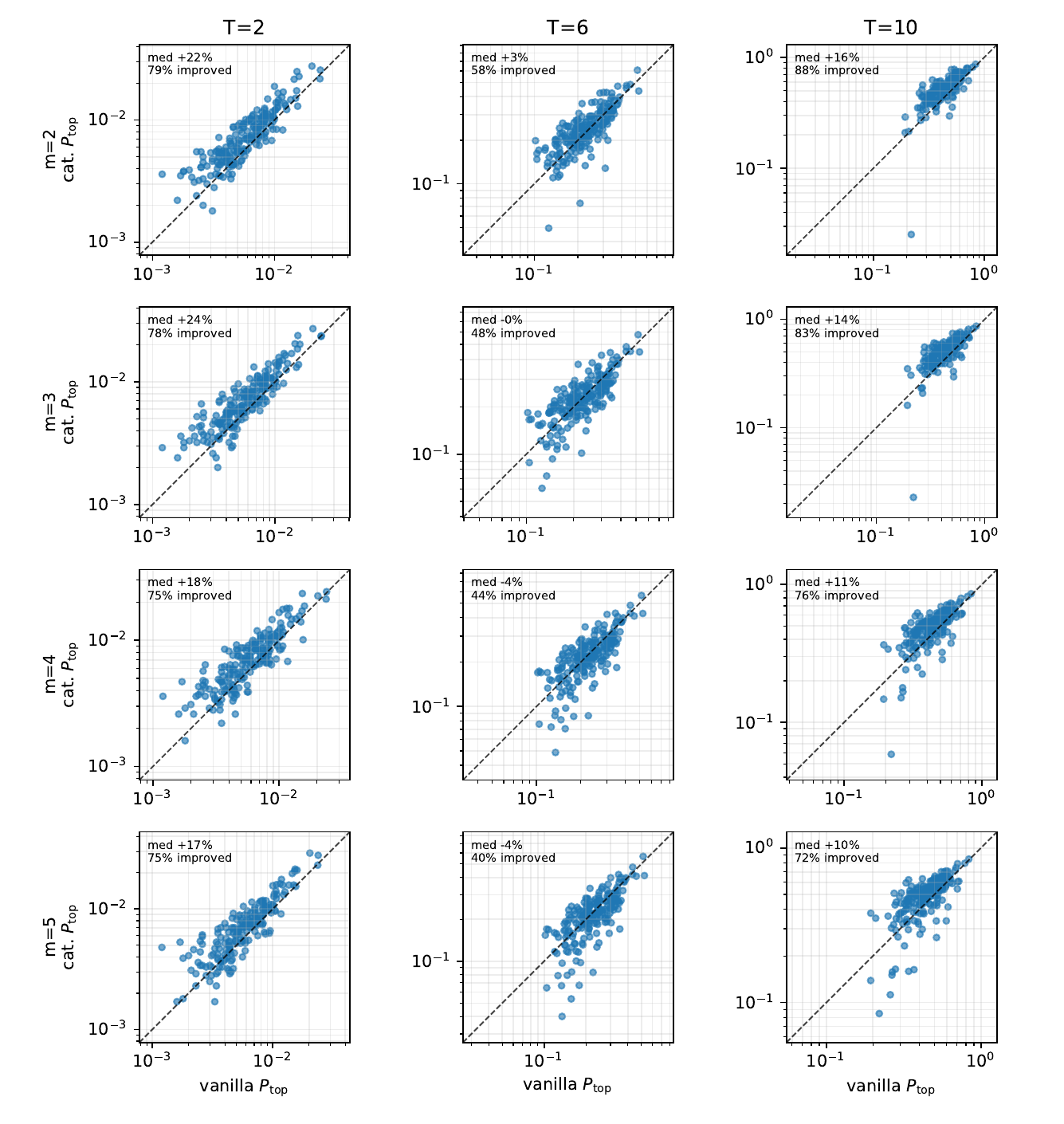}
\caption{\textbf{Fully connected instances, top-band mass.} Catalyzed versus uncatalyzed $P_{0.95}$ on $200$ Sherrington--Kirkpatrick-type instances ($n=20$). Rows: path orders $m=2$--$5$. Columns: sweep times $T=2,6,10$.}
\label{fig:sm-sk-energy}
\end{figure*}

\begin{figure*}[p]
\centering
\includegraphics[width=\linewidth]{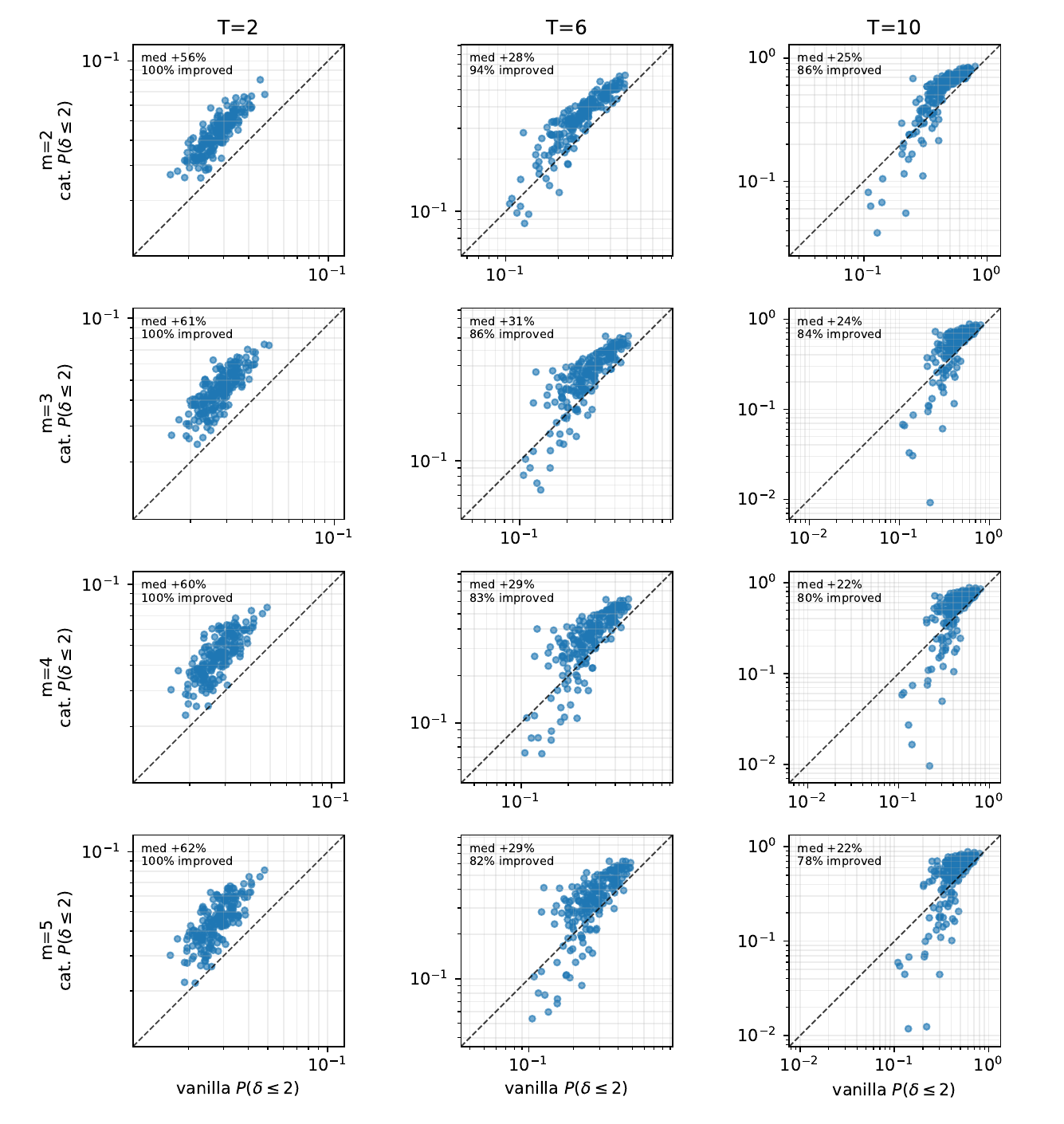}
\caption{\textbf{Fully connected instances, near-solution Hamming mass.} As in Fig.~\ref{fig:sm-sk-energy}, for $P(\delta\le2)$.}
\label{fig:sm-sk-hamming}
\end{figure*}

\begin{figure*}[p]
\centering
\includegraphics[width=\linewidth]{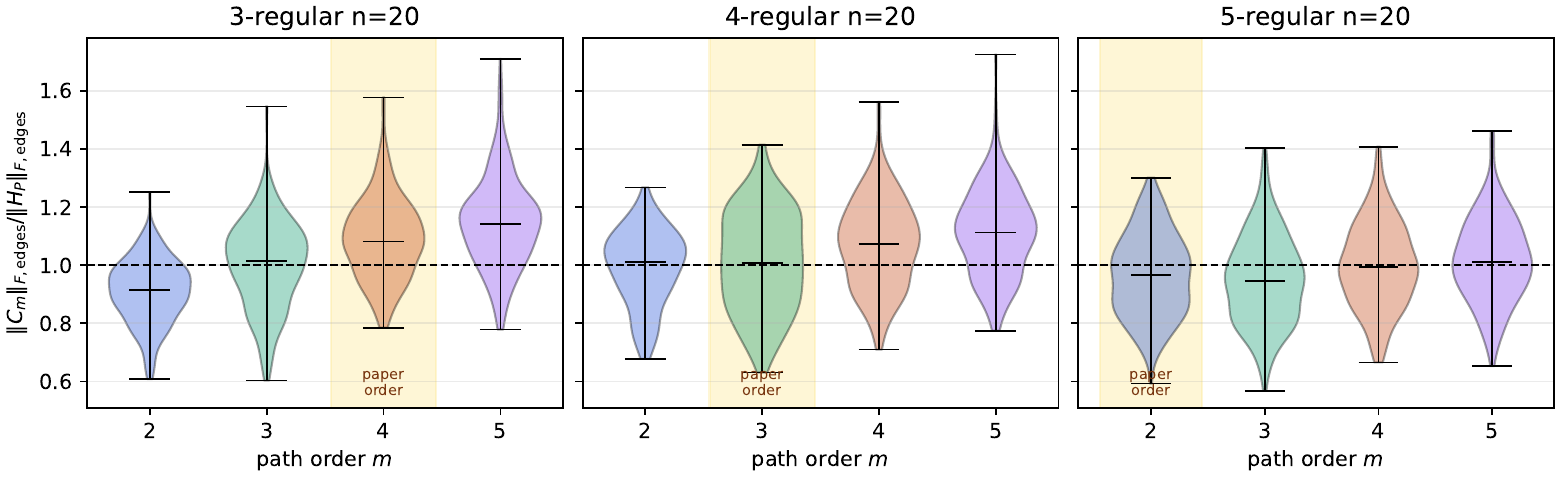}
\caption{\textbf{Coupling-budget comparison.} Distribution across instances of the Frobenius-norm ratio $\lVert C^{(m)}\rVert_F/\lVert H_P\rVert_F$ at unit peak coupling ($\max_{i<j}|W_{ij}|=1$ for each operator), for path orders $m=2$--$5$ on the three graph families. Shaded columns mark the orders used in the main text. At those orders the median ratio is $1.08$, $1.01$, and $0.97$, so the catalysts' advantage over the reinforced-problem control is obtained at essentially matched total coupling weight.}
\label{fig:sm-frobenius}
\end{figure*}
\begin{table*}[p]
\caption{\textbf{Median paired gains over the uncatalyzed anneal} (with fraction of instances improved), for fixed orders $m=4,3,2$ on 3-, 4-, and 5-regular families.}
\label{tab:sm-vanilla}
\begin{ruledtabular}
\begin{tabular}{llcccc}
Family & $T$ & $P_{0.90}$ & $P_{0.95}$ & $P_{0.98}$ & $P(\delta\le2)$ \\
\hline
3-regular & 2  & $+382\%$ (100\%) & $+487\%$ (100\%) & $+588\%$ (100\%) & $+277\%$ (100\%) \\
          & 6  & $+229\%$ (100\%) & $+351\%$ (99\%)  & $+416\%$ (98\%)  & $+239\%$ (94\%)  \\
          & 10 & $+87\%$ (98\%)   & $+149\%$ (93\%)  & $+190\%$ (85\%)  & $+148\%$ (80\%)  \\
4-regular & 2  & $+270\%$ (100\%) & $+318\%$ (100\%) & $+350\%$ (100\%) & $+216\%$ (100\%) \\
          & 6  & $+125\%$ (100\%) & $+184\%$ (100\%) & $+221\%$ (94\%)  & $+165\%$ (94\%)  \\
          & 10 & $+51\%$ (100\%)  & $+87\%$ (96\%)   & $+111\%$ (88\%)  & $+96\%$ (86\%)   \\
5-regular & 2  & $+193\%$ (100\%) & $+216\%$ (100\%) & $+241\%$ (100\%) & $+153\%$ (100\%) \\
          & 6  & $+94\%$ (100\%)  & $+133\%$ (100\%) & $+164\%$ (95\%)  & $+115\%$ (99\%)  \\
          & 10 & $+39\%$ (100\%)  & $+57\%$ (95\%)   & $+75\%$ (90\%)   & $+68\%$ (88\%)   \\
\end{tabular}
\end{ruledtabular}
\end{table*}

\begin{table}[p]
\caption{\textbf{Median paired gains over the reinforced-problem baseline} $C=H_P$ (with fraction improved), at $\rho_0=0.95$ and in near-solution Hamming mass.}
\label{tab:sm-baseline}
\begin{ruledtabular}
\begin{tabular}{llcc}
Family & $T$ & $P_{0.95}$ & $P(\delta\le2)$ \\
\hline
3-regular & 2  & $+68\%$ (96\%) & $+98\%$ (98\%)  \\
          & 6  & $+77\%$ (92\%) & $+102\%$ (86\%) \\
          & 10 & $+41\%$ (86\%) & $+64\%$ (73\%)  \\
4-regular & 2  & $+38\%$ (92\%) & $+76\%$ (99\%)  \\
          & 6  & $+33\%$ (88\%) & $+63\%$ (88\%)  \\
          & 10 & $+24\%$ (84\%) & $+43\%$ (79\%)  \\
5-regular & 2  & $+21\%$ (81\%) & $+50\%$ (100\%) \\
          & 6  & $+19\%$ (82\%) & $+40\%$ (93\%)  \\
          & 10 & $+14\%$ (82\%) & $+28\%$ (80\%)  \\
\end{tabular}
\end{ruledtabular}
\end{table}

\newpage
\section*{Numerical dynamics and normalization}
\label{app:numerics}
All dynamical data were generated by direct state-vector simulation of the time-dependent Hamiltonian
\begin{equation}
H(t) = (1-s)H_D + sH_P + s(1-s)C^{(m)},\qquad
s=t/T ,
\end{equation}
with transverse-field driver
\begin{equation}
    H_D=-\sum_i X_i,
\end{equation}
problem Hamiltonian
\begin{equation}
    H_P=\sum_{i<j} J_{ij} Z_i Z_j ,
\end{equation}
and catalyst Hamiltonian
\begin{equation}
    C^{(m)}=\sum_{i<j} C^{(m)}_{ij} Z_i Z_j .
\end{equation}
The initial state was \(|+\rangle^{\otimes n}\). Time evolution was computed with the QiliSDK \cite{qilisdk} state-vector backend using its fixed-step \texttt{TrotterizedSchedule} and \texttt{DigitalPropagation} routines, followed by computational-basis sampling. The production data use time step \(\Delta t=0.1\) and \(10^4\) measurement samples per method, instance, and annealing time. No adaptive ODE tolerance is used; the numerical accuracy parameter is the fixed product-formula step size \(\Delta t\). The scripts store \(\Delta t\), \(T\), and the sampling count in every output file, and the same instances can be rerun at smaller \(\Delta t\) for step-size checks.The problem couplings are peak-normalized before any dynamics:
\begin{equation}
J_{ij}\leftarrow \frac{J_{ij}}{\max_{a<b}|J_{ab}|},
\qquad \max_{i<j}|J_{ij}|=1 .
\end{equation}

Reverse path orientations are not counted separately: a path and its reversal are identified, and only one canonical orientation contributes. To remove the component of the catalyst parallel to the problem Hamiltonian, we orthogonalize \(C^{(m)}_0\) against \(J\) using the upper-triangular edge inner product

\begin{equation}
\langle A,B\rangle_E=\sum_{i<j} A_{ij}B_{ij}.
\end{equation}
The projected catalyst is
\begin{equation}
\widetilde C^{(m)} = C^{(m)}_0 -
\frac{\langle C^{(m)}_0,J\rangle_E}{\langle J,J\rangle_E}J .
\end{equation}
It is then peak-normalized as
\begin{equation}
C^{(m)}_{ij}
= \frac{\widetilde C^{(m)}_{ij}} {\max_{a<b}|\widetilde C^{(m)}_{ab}|}.
\end{equation}
All Hamming catalysts in the main comparison therefore also satisfy \(\max_{i<j}|C^{(m)}_{ij}|=1\), and the comparison with \(C=H_P\) is at equal peak coupling. \\

Random \(k\)-regular instances are sampled by a configuration-model stub matching procedure with rejection of self-loops and parallel edges, until a simple \(k\)-regular graph is obtained. The required conditions are \(n\ge k+1\) and \(nk\) even. We do not additionally condition on graph connectedness. Couplings are then assigned independently on the graph support from the specified distribution, and finally rescaled to unit peak coupling as above. Complete-graph instances use all pairs \(i<j\) as the support, with independent signed continuous couplings before the same peak normalization. \\

Ground states are found by exhaustive enumeration for the simulated sizes. All computational-basis states with energy \(E\le E_{\min}+10^{-9}\) are retained as ground states. Hamming-distance observables are computed using the distance to the nearest member of this ground-state manifold, which automatically includes the global spin-flip degeneracy and any additional accidental degeneracies. Relative gains are computed instance by instance as

\begin{equation}
R_i = \frac{P_i^{\mathrm{cat}}-P_i^{\mathrm{base}}}
     {P_i^{\mathrm{base}}}.
\end{equation}
If \(P_i^{\mathrm{base}}=0\), the ratio is treated as undefined and omitted from ratio summaries; this does not occur for the main \(\rho_0=0.95\) analyses, but occurs in three auxiliary short-time comparisons for 3-regular instances at \(\rho_0=0.98\). Very small but nonzero baseline probabilities are not clipped in the numerical summaries. Clipping to a floor is used only for log-scale visualization of scatter plots, not for computing medians, quantiles, or fractions of improved instances. \\

\emph{Data availability.}
The instance ensembles, final-state distributions, and per-instance statistics underlying Figs.~S1--S9 and Tables~S1--S2, together with the analysis and plotting scripts, are openly available at
\url{https://github.com/qilimanjaro-tech/diagonal-catalysts-data}.

\end{document}